%% file: main.tex
\documentclass[final]{l4dc2026}
\input{Sources/header.tex}
\input{Sources/colors.tex}

\newif\ifcompile
\compiletrue

\title[Policy Optimization for Unknown Systems using Differentiable MPC]{Policy Optimization for Unknown Systems using Differentiable \\ Model Predictive Control}

\author{%
 \Name{Riccardo Zuliani} \Email{rzuliani@ethz.ch}\\
 \addr Automatic Control Laboratory, ETH Z\"urich
 \AND
 \Name{{Efe C.} Balta} \Email{efe.balta@inspire.ch}\\
 \addr Inspire AG, Z\"urich and Automatic Control Laboratory, ETH Z\"urich
 \AND
 \Name{John Lygeros} \Email{jlygeros@ethz.ch}\\
 \addr Automatic Control Laboratory, ETH Z\"urich
}

\begin{document}

\maketitle

\begingroup
\renewcommand{\thefootnote}{}% Empty marker
\footnotetext{\hspace*{0pt}This work was supported as a part of NCCR Automation, a National Centre of Competence in Research, funded by the Swiss National Science Foundation (grant number 51NF40\_225155).}
\endgroup

\begin{abstract}%
Model-based policy optimization often struggles with inaccurate system dynamics models, leading to suboptimal closed-loop performance. This challenge is especially evident in Model Predictive Control (MPC) policies, which rely on the model for real-time trajectory planning and optimization. We introduce a novel policy optimization framework for MPC-based policies combining differentiable optimization with zeroth-order optimization. Our method combines model-based and model-free gradient estimation approaches, achieving faster transient performance compared to fully data-driven approaches while maintaining convergence guarantees, even under model uncertainty. We demonstrate the effectiveness of the proposed approach on a nonlinear control task involving a 12-dimensional quadcopter model.
\end{abstract}


\section{Introduction}\label{section:introduction}
\input{Chapters/introduction.tex}

\section{Preliminaries}\label{section:preliminaries}
\input{Chapters/preliminaries.tex}

\section{Problem Formulation}\label{section:formulation}
\input{Chapters/problem_formulation.tex}

\section{Model-based Gradient Estimation}\label{section:model_based}
\input{Chapters/model_based.tex}

\section{Model-free Gradient Estimation}\label{section:model_free}
\input{Chapters/model_free.tex}

\section{Proposed Algorithm}\label{section:proposed}
\input{Chapters/proposed_algorith.tex}

\section{Simulation Results}\label{section:simulation}
\input{Chapters/simulation_results.tex}

\section{Conclusions and Future Work}\label{section:conclusion}
\input{Chapters/conclusion.tex}

\bibliography{Sources/ref.bib}

\appendix

\section{Functions Definable in an o-minimal Geometry}
\label{section:preliminaries:definable_functions}
\input{Chapters/definability.tex}

\section{Differentiating Solutions of Optimization Problems}
\label{section:preliminaries:differentiating_optimization_problems}
\input{Chapters/appendix_mpc_path_diff.tex}

\section{Ablation Studies}
\label{section:ablation}
\input{Chapters/ablation.tex}

\end{document}

%% file: Sources/colors.tex
\definecolor{red1}{HTML}{E3120B} % #E3120B
\definecolor{red2}{HTML}{CC100A} % #CC100A
\definecolor{red3}{HTML}{F6423C} % #F6423C
\definecolor{red4}{HTML}{FEE7E7} % #FEE7E7

\definecolor{backgroundred1}{HTML}{E1D0D0} % #E1D0D0
\definecolor{backgroundred2}{HTML}{EBE0E0} % #EBE0E0
\definecolor{backgroundred3}{HTML}{F5F0EF} % #F5F0EF

\definecolor{blue1}{HTML}{141F52} % #141F52
\definecolor{blue2}{HTML}{1F2E7A} % #1F2E7A
\definecolor{blue3}{HTML}{2E45B8} % #2E45B8
\definecolor{blue4}{HTML}{475ED1} % #475ED1
\definecolor{blue5}{HTML}{D6DBF5} % #D6DBF5
\definecolor{blue6}{HTML}{EBEDFA} % #EBEDFA

\definecolor{backgroundblue1}{HTML}{D0D3E1} % #D0D3E1
\definecolor{backgroundblue2}{HTML}{E0E2EB} % #E0E2EB
\definecolor{backgroundblue3}{HTML}{EFF0F5} % #EFF0F5

\definecolor{green1}{HTML}{4C9C16} % #4C9C16
\definecolor{green2}{HTML}{62C91D} % #62C91D
\definecolor{green3}{HTML}{7BE236} % #7BE236
\definecolor{green4}{HTML}{E2F9D2} % #E2F9D2
\definecolor{green5}{HTML}{F0FCE9} % #F0FCE9

\definecolor{yellow1}{HTML}{F9C31F} % #F9C31F
\definecolor{yellow2}{HTML}{FBD051} % #FBD051
\definecolor{yellow3}{HTML}{FCDE83} % #FCDE83
\definecolor{yellow4}{HTML}{FEF2CD} % #FEF2CD
\definecolor{yellow5}{HTML}{FEF8E6} % #FEF8E6

\definecolor{orange1}{HTML}{F97A1F} % #F97A1F
\definecolor{orange2}{HTML}{FB9851} % #FB9851
\definecolor{orange3}{HTML}{FCB583} % #FCB583
\definecolor{orange4}{HTML}{FEE1CD} % #FEE1CD
\definecolor{orange5}{HTML}{FEF0E6} % #FEF0E6

\definecolor{gray1}{HTML}{0D0D0D} % #0D0D0D
\definecolor{gray2}{HTML}{1A1A1A} % #1A1A1A
\definecolor{gray3}{HTML}{333333} % #333333
\definecolor{gray4}{HTML}{595959} % #595959
\definecolor{gray5}{HTML}{666666} % #666666
\definecolor{gray6}{HTML}{B3B3B3} % #B3B3B3
\definecolor{gray7}{HTML}{D9D9D9} % #D9D9D9
\definecolor{gray8}{HTML}{F2F2F2} % #F2F2F2
\definecolor{gray9}{HTML}{FFFFFF} % #FFFFFF

\colorlet{main_red}{red2}
\colorlet{main_blue}{blue2}
\colorlet{main_orange}{orange2}
\colorlet{main_gray}{gray5}
\colorlet{main_yellow}{yellow2}
\colorlet{main_black}{gray2}

\colorlet{parameters_flowchart}{main_red}
\colorlet{system_flowchart}{main_blue}
\colorlet{estimation_flowchart}{main_orange}
\colorlet{zeroth_flowchart}{main_gray}
\colorlet{parameters_background_flowchart}{main_red!10}
\colorlet{system_background_flowchart}{main_blue!10}
\colorlet{estimation_background_flowchart}{main_orange!10}
\colorlet{zeroth_background_flowchart}{main_gray!10}

\pgfplotsset{solid/.style={
    very thick,
}}
\pgfplotsset{dashed/.style={
    very thick,
    dash pattern=on 5pt off 2pt,
}}
\pgfplotsset{dashdot/.style={
    very thick,
    dash pattern=on 5.25pt off 1.75pt on 1.75pt off 1.75pt
}}
\pgfplotsset{marked/.style={
    mark=square,
    mark repeat={10},
    mark size = 1,
    mark options={rotate=45,solid},
    thick,
}}

\pgfplotsset{x_best/.style={dashed, marked, main_blue!40!gray4}}
\pgfplotsset{y_best/.style={dashed, marked, main_orange!40!gray4}}
\pgfplotsset{z_best/.style={dashed, marked, main_red!40!gray4}}
\pgfplotsset{phi_best/.style={dashed, marked, main_blue!40!gray4}}
\pgfplotsset{theta_best/.style={dashed, marked, main_orange!40!gray4}}
\pgfplotsset{psi_best/.style={dashed, marked, main_red!40!gray4}}

\pgfplotsset{x_ours/.style={solid, main_blue}}
\pgfplotsset{y_ours/.style={solid, main_orange}}
\pgfplotsset{z_ours/.style={solid, main_red}}
\pgfplotsset{phi_ours/.style={solid, main_blue}}
\pgfplotsset{theta_ours/.style={solid, main_orange}}
\pgfplotsset{psi_ours/.style={solid, main_red}}

\pgfplotsset{bl25/.style={solid, main_red}}
\pgfplotsset{bl50/.style={solid, main_orange}}
\pgfplotsset{bl75/.style={solid, main_yellow}}
\pgfplotsset{zo/.style={dashdot, main_blue}}
\pgfplotsset{model/.style={semithick, dashdot, main_gray, opacity = 0.35}}

\colorlet{grid_gray}{gray7}
\colorlet{connect_gray}{main_gray}

%% file: Chapters/introduction.tex
Policy optimization is the problem of designing a control policy that minimizes a prescribed performance objective,
typically by searching over a parameterized policy class \citep{sutton2002reinforcement}. 
% \emph{Model-based} RL methods use a model of the dynamics to guide policy updates, improving sample efficiency \cite{atkeson1997comparison}.
A growing line of work studies policy optimization problems where the policy is a model predictive controller (MPC) \citep{amos2018differentiable,gros2019data,agrawal2020learning,drgovna2022differentiable,zuliani2023bp,zuliani2024closed}. 
MPC-based policies generate predictions of future state trajectories using a system model and naturally incorporate constraints into their decision-making process, offering stronger safety guarantees and greater interpretability compared to model-free approaches.

Existing MPC-based policy optimization schemes typically assume that the dynamics model is exact, 
and convergence results rely on this assumption \citep{zuliani2023bp}; 
obtaining convergence guarantees when the model is inaccurate remains an open problem. 
In this paper, we propose a novel policy optimization algorithm with convergence guarantees inspired by \cite{he2024gray} blending model-based and zeroth-order gradient information that is robust to inexact models. 
A key feature of our method is that it can smoothly trade off between model-based and zeroth-order components, putting more weight on the model whenever it is trusted and relying more on model-free information otherwise.
To handle the nonsmoothness of MPC policies, we leverage the tools of \cite{bolte2021conservative}.
We validate our approach on a 12-dimensional nonlinear quadcopter.

\textbf{Related work:}
Zeroth-order optimization addresses the problem of minimizing an objective function when exact gradient information is unavailable.
The foundations of the approach used in this paper trace back to the seminal work of \cite{flaxman2004online}, which introduced a smoothing-based approximation technique enabling gradient-free optimization for possibly nonsmooth functions.
Subsequent research extended these ideas to convex settings: see \cite{duchi2012randomized} and \cite{nesterov2017random} for one-point gradient estimators, or \cite{shamir2017optimal} for a two-point estimator.
More recently, \cite{lin2022gradient} generalized the two-point approach to nonconvex problems, demonstrating its effectiveness beyond the convex regime.
Closest to our work is \cite{he2024gray}, which combines a one-point zeroth-order estimator with model-based gradient information to improve convergence speed, a direction we further build upon in this paper.

\textbf{Notation}: $\mathcal{X}^n$ denotes the $n$-fold Cartesian product of the set $\mathcal{X}$. Given a path-differentiable function $g$ of two arguments $x$ and $y$, $\J_{g,x}$ and $\J_{g,y}$ denote the projection of the conservative Jacobian $\J_g$ \rz{(defined in Section~\ref{section:preliminaries:conservative_jacobians})} onto the $x$ and $y$ entries. $\mathbb{B}$ and $\mathbb{S}$ are the unit ball and sphere in the Euclidean norm. $U(\mathbb{B})$ and $U(\mathbb{S})$ denote uniform distributions over $\mathbb{B}$ and $\mathbb{S}$. $\mathcal{N}_\mathcal{X}$ is the Clarke tangent cone of the set $\mathcal{X}$. $\operatorname*{dist}_p(a,B)$ is the distance between point $a$ and set $B$ in the $p$-norm. We use $\mathbf{0}_n, \mathbf{1}_n \in \R^n$ to denote the vector of zeros and ones, respectively.

\textbf{Outline:}
We graphically outline the rest of this paper graphically in Figure~\ref{fig:outline}.

\begin{figure}[hb!]
    \centering
    \ifcompile
      \input{Figures/flowchart/flowchart_figure.tex}
    \else
      \includegraphics{Figures/flowchart/flowchart.pdf}
    \fi
    \caption{Graphical outline of the paper.}
    \label{fig:outline}
\end{figure}
\vspace{-1cm}

%% file: Figures/flowchart/flowchart_figure.tex
\hypersetup{hidelinks}
\begin{tikzpicture}[
    node distance=1.2cm,
    auto,
    font=\small,
    % Styles
    block/.style={rectangle, draw, thick, rounded corners, minimum width=1.5cm, minimum height=1.1cm, text centered, font=\small},
    start/.style={text centered, font=\small\bfseries, text width=1.8cm},
    arrow/.style={-Stealth, thick},
    rowlabel/.style={font=\small\bfseries, text width=1.5cm, align=center},
    sum/.style={circle, draw, thick, minimum size=0.5cm, text centered, font=\small},
    mult/.style={rectangle, draw, thick, minimum size=0.5cm, text centered, font=\small} % New style for multiplication blocks
]

% ==================== ROW 1: MODEL-BASED ESTIMATE ====================

\node (label1) [rowlabel] {Run 1};
\node (theta1) [start, right=0.1cm of label1, text=parameters_flowchart] {$\theta_k$};
\node (cl1) [block, right=0.75cm of theta1, draw=system_flowchart, text=system_flowchart, align=center] {Closed\\Loop};
\node (model) [block, right=3.2cm of cl1, draw=estimation_flowchart, text=estimation_flowchart, align=center] {Model\\Based};

\draw[arrow, parameters_flowchart] (theta1) -- (cl1);
% Label named 'data1' for vertical alignment
\draw[arrow,system_flowchart] (cl1) -- (model) node[pos=0.7, above=0.1cm, name=data1, font=\small] {};

% add label
\node[system_flowchart,anchor=west,font=\small] (data1label) at ([yshift=0.1cm,xshift=-2.2cm]data1.center) {$x(\theta_k), u(\theta_k)$};

% ==================== ROW 2: ZEROTH-ORDER ESTIMATE ====================

% Vertical separation
\coordinate (row2_y) at ($(cl1.south) + (0cm,-1.8cm)$);

\node (label2) [rowlabel, at={(label1 |- row2_y)}] {Run 2};
\node (theta2) [start, at={(theta1 |- row2_y)}, text=parameters_flowchart] {$\theta_k + \delta v_k$};

% Precisely align CL2 and ZO under CL1 and Model
\node (cl2) [block, at={(cl1 |- row2_y)}, draw=system_flowchart, text=system_flowchart, align=center] {Closed\\Loop};
\node (zo) [block, at={(model |- row2_y)}, draw=zeroth_flowchart, text=zeroth_flowchart, align=center] {Zeroth\\Order};

% Position Sum exactly under Row 1 data label
\node (sum) [sum, at={(data1 |- row2_y)}, system_flowchart] {};
\node[text=system_flowchart] at (sum.center){$+$};

% Label precisely below the sum node
\node[below=1.4cm of data1label, xshift=-0.05cm, align=center, font=\small, system_flowchart] (sumlabel) {$x(\theta_k \!+\! \delta v_k)$ \\[-0.2em] $u(\theta_k \!+\! \delta v_k)$};

% Paths Row 2
\draw[arrow, parameters_flowchart] (theta2) -- (cl2);
\draw[arrow, system_flowchart] (cl2) -- (sum);
\draw[arrow, system_flowchart] (sum) -- (zo);

% ==================== CROSS-ROW CONNECTION ====================

% Vertical arrow from Row 1 data to Sum North
\draw[arrow, system_flowchart] ([yshift=-0.1cm]data1.south) -- (sum.north) node[right, pos=0.8] {$-$};
\fill[system_flowchart] ([yshift=-0.1cm]data1.south) circle (2pt);

% ==================== AGGREGATION ====================

% Multiplication block for Model-Based (Changed from 'sum' to 'mult')
\node[parameters_flowchart] (mult1) [mult, right=0.8cm of model] {};
\node[parameters_flowchart] at (mult1.center) {$\times$};

% Multiplication block for Zeroth-Order (Changed from 'sum' to 'mult')
\node[parameters_flowchart] (mult2) [mult, at={(mult1 |- zo)}] {};
\node[parameters_flowchart] at (mult2.center) {$\times$};

% Labels for the multiplication blocks aligned vertically on the outside
\node[above=0.1cm of mult1, font=\small, text=parameters_flowchart] {$\eta_k$};
\node[below=0.1cm of mult2, xshift=0.1cm, font=\small, text=parameters_flowchart] {$1-\eta_k$};

% Final Summation block aligned precisely on the same vertical axis
\node[parameters_flowchart] (sum_final) [sum, at={($(mult1)!0.5!(mult2)$)}] {};
\node[parameters_flowchart] at (sum_final.center){$+$};

% Paths connecting Model and ZO to the multipliers
\draw[arrow, estimation_flowchart] (model) -- (mult1);
\draw[arrow, zeroth_flowchart] (zo) -- (mult2);

% Vertical paths connecting the multipliers to the final sum block
\draw[arrow, parameters_flowchart] (mult1.south) -- (sum_final.north);
\draw[arrow, parameters_flowchart] (mult2.north) -- (sum_final.south);

% Final Gradient Output Arrow
\draw[arrow, parameters_flowchart] (sum_final.east) -- ++(0.7cm,0) node[right, font=\small, text=parameters_flowchart] (output) {$\J_{\mathcal{C}^\delta}(\theta_k)$};

% ==================== BACKGROUNDS & TITLES ====================

% Declare background layer
\pgfdeclarelayer{background}
\pgfsetlayers{background,main}

\begin{pgfonlayer}{background}
    % 1. Calculate boundaries (midpoints between specific nodes)
    \coordinate (L_edge) at ($(label1.east)!0.5!(theta1.west)$);
    \coordinate (B1)     at ($(theta1.east)!0.5!(cl1.west)$);
    \coordinate (B2)     at ($(sum.east)!0.5!(model.west)$);
    \coordinate (B3)     at ($(model.east)!0.5!(mult1.west)$);
    \coordinate (R_edge) at ($(output.east) + (0.2cm, 0)$);

    % 2. Calculate vertical limits
    \coordinate (Top_edge)    at ($(label1.north) + (0, 1.1cm)$);
    \coordinate (Bottom_edge) at ($(sumlabel.south) + (0, -1cm)$);
    \coordinate (Mid_Y)       at ($(sum_final.center) + (0, 0.2cm)$);

    % 3. Draw the background rectangles with rounded corners and slight gaps
    \fill[parameters_background_flowchart, rounded corners=5mm] 
        (L_edge |- Bottom_edge) rectangle ([xshift=-0.1cm]B1 |- Top_edge);
        
    \fill[system_background_flowchart, rounded corners=5mm]   
        ([xshift=0.1cm]B1 |- Bottom_edge) rectangle ([xshift=-0.1cm]B2 |- Top_edge);
    
    % Split the 3rd column (with vertical and horizontal gaps)
    \fill[estimation_background_flowchart, rounded corners=5mm] 
        ([xshift=0.1cm, yshift=0.1cm]B2 |- Mid_Y) rectangle ([xshift=-0.1cm]B3 |- Top_edge);
    \fill[zeroth_background_flowchart, rounded corners=5mm]     
        ([xshift=0.1cm]B2 |- Bottom_edge) rectangle ([xshift=-0.1cm, yshift=-0.1cm]B3 |- Mid_Y);
    
    \fill[parameters_background_flowchart, rounded corners=5mm] 
        ([xshift=0.1cm]B3 |- Bottom_edge) rectangle (R_edge |- Top_edge);
\end{pgfonlayer}

\coordinate (Mid1) at ($(L_edge)!0.5!(B1)$);
\coordinate (Mid2) at ($(B1)!0.5!(B2)$);
\coordinate (Mid3) at ($(B2)!0.5!(B3)$);
\coordinate (Mid4) at ($(B3)!0.5!(R_edge)$);

% --- Column Titles ---
% Use the pre-calculated midpoints with |- for clean compilation
\node[text=parameters_flowchart, font=\small, align=center, anchor=north]
    at (Mid1 |- Top_edge) {\textbf{Sec. \ref{section:proposed}}};
\node[text=system_flowchart, font=\small, align=center, anchor=north]
    at (Mid2 |- Top_edge) {\textbf{Sec. \ref{section:formulation}}};
    
% Top half title
\node[text=estimation_flowchart, font=\small, align=center, anchor=north]
    at (Mid3 |- Top_edge) {\textbf{Sec. \ref{section:model_based}}};
    
% Bottom half title
\node[text=zeroth_flowchart, font=\small, align=center, anchor=south, yshift=1.65cm]
    at (Mid3 |- Bottom_edge) {\textbf{Sec. \ref{section:model_free}}};
    
\node[text=parameters_flowchart, font=\small, align=center, anchor=north]
    at (Mid4 |- Top_edge) {\textbf{Sec. \ref{section:proposed}}};

\end{tikzpicture}

%% file: Chapters/preliminaries.tex
\rz{Since MPC-based policy optimization is a fundamentally nonsmooth problem, we adopt the framework of definable functions to ensure existence of well-defined generalized Jacobians. We formally define the generalized Jacobian in Section~\ref{section:preliminaries:definable_functions} and report the definition of definability in Appendix~\ref{section:preliminaries:conservative_jacobians}. The derivation of generalized Jacobians for MPC solution maps is reported in Appendix~\ref{section:preliminaries:differentiating_optimization_problems}. Finally, to account for both the nonsmooth landscape and unknown dynamics, we characterize the convergence of our algorithm using generalized critical points, defined in Section~\ref{section:preliminaries:critical_points}.}

\subsection{Conservative Jacobians}
\label{section:preliminaries:conservative_jacobians}
The notion of Conservative Jacobians,
introduced in \cite{bolte2021conservative},
extends the concept of derivatives to locally Lipschitz, 
almost everywhere differentiable functions.
\begin{definition}[\cite{bolte2021conservative}]
The map $\mathcal{J}_{\phi}:\R^n\rightrightarrows\R^{m \times n}$ is a conservative Jacobian of the locally Lipschitz function $\phi:\R^n\to\R^m$ if it is nonempty-valued, outer semicontinuous, locally bounded, and for all absolutely continuous curves $x:[0,1]\to\R^n$ and almost all $t\in[0,1]$
\begin{align}\label{eq:CJ_chain_rule}
\frac{\text{d}}{\text{d}t}\phi(x(t)) = \langle v,\dot{x}(t) \rangle ,~ \forall v\in \mathcal{J}_{\phi}(x(t)).
\end{align}
\end{definition}
Excluding a set of measure zero, $\J_\phi$ coincides almost everywhere with the gradient $\nabla \phi$.
Unlike other nonsmooth Jacobians, such as Clarke Jacobians, conservative Jacobians obey the chain rule of differentiation and a nonsmooth version of the implicit function theorem \citep{bolte2021nonsmooth},
making them particularly suitable for the sensitivity analysis of solution maps of optimization problems,
which are typically nonsmooth functions implicitly defined by optimality conditions.

A function admitting a conservative Jacobian is called \emph{path-differentiable}. As demonstrated in Lemma 3, \cite{bolte2021nonsmooth}, locally Lipschitz definable functions are always path-differentiable.

\subsection{Stationary points in nonsmooth optimization}
\label{section:preliminaries:critical_points}

The \emph{Clarke Jacobian} of a locally Lipschitz function 
$ f:\R^n \to \R $ 
is the outer-semicontinuous map
\begin{align*}
\partial_c f(x) := \operatorname*{co} \{ g \in \R^n : g = \lim_{y \to x} \nabla f(y),~ y \in D_f \},
\end{align*}
where $ D_f \subset \R^n $ is the full-measure set on which $ f $ is differentiable, 
and $ \operatorname*{co} $ denotes the convex hull. 
There is a close connection between Clarke and conservative Jacobians \citep{bolte2021conservative}, where $\partial_c f(x) \subseteq \operatorname*{co}\J_f(x)$ for all $x$, and $\partial_c f(x) = \J_f(x)$ for almost every $x$.
% \begin{align*}
% \partial_c f(x) \subseteq \operatorname*{co}\J_f(x) ~~ \text{for all } x, \qquad \partial_c f(x) = \J_f(x) ~~ \text{for almost every } x.
% \end{align*}

If $x$ is a local minimizer of $f$, then $0 \in \partial_c f(x)$ 
(and similarly $0 \in \operatorname*{co}\J_f(x)$).
Hence, in nonsmooth optimization, one typically searches for 
\emph{Clarke stationary points}, that is, points $x$ with $0 \in \partial_c f(x)$.
In our setting, identifying Clarke stationary points is impossible to due lack of full information regarding the objective function.  
Consequently, we adopt a weaker notion of stationarity, namely, that of a 
\emph{Goldstein $\delta$-critical point}, defined as any $x$ such that 
$0 \in \partial_\delta f(x)$, where
\begin{align*}
\partial_\delta f(x) := \operatorname*{co} \left\{ {\textstyle \bigcup_{y \in \delta\mathbb{B}}} \partial_c f(x+y) \right\}
\end{align*}
is the Goldstein $\delta$-subdifferential of $f$.
As shown in \cite{zhang2020complexity}, 
$\lim_{\delta \downarrow 0} \partial_\delta f(x) = \partial_c f(x)$,
making Goldstein $\delta$-stationarity a meaningful optimality condition 
for nonsmooth problems.

For constrained problems, such as minimizing $f(x)$ subject to $h(x) = 0$, $x\in \mathcal{X}$, 
we consider a generalized stationarity concept adapted from \cite{grimmer2025goldstein}.
A point $x$ is said to satisfy the \emph{Goldstein Fritz-John $\delta$-critical condition} 
if there exist multipliers $\lambda_0, \lambda_1 \ge 0$ such that $0 \in \lambda_0 \partial_\delta f(x) + \lambda_1 \partial_\delta h(x) + \mathcal{N}_\mathcal{X}(x)$ and $\lambda_0 + \lambda_1 = 1$.
% \begin{align*}
% 0 \in \lambda_0 \partial_\delta f(x) + \lambda_1 \partial_\delta h(x) + \mathcal{N}_\mathcal{X}(x),~~ \lambda_0 + \lambda_1 = 1.
% \end{align*}
This condition parallels the classical Fritz-John optimality conditions,
with the gradients replaced by Goldstein $\delta$-subdifferentials.
If $\lambda_0 > 0$, the point is called a \emph{Goldstein KKT $\delta$-critical point}.

%% file: Chapters/problem_formulation.tex
Consider an unknown system evolving over $t\in\Z_{[0,T]}$ from an initial condition $x_0\in\R^{n_x}$
\begin{subequations}
\begin{align}
x_{t+1} &= f(x_t,u_t), \label{eq:problem_formulation:system_true} \\
u_t &= \operatorname*{MPC}(x_t,\theta), \label{eq:problem_formulation:mpc}
\end{align}
\end{subequations}
where $\theta\in\Theta$ is a tunable parameter that determines the behavior of the policy and $\Theta \subset \R^{n_\theta}$ is a parameter set.
Despite not knowing the dynamics in \eqref{eq:problem_formulation:system_true} exactly, 
we assume existence of a model $g:\R^{n_x} \times \R^{n_u} \to \R^{n_x}$ such that 
$g(x,u) \approx f(x,u)$ for all $x$ and $u$.
\rz{For example, $g$ may be a linear model obtained via system identification of the nonlinear plant $f$,
or a first-principles model that neglects higher-order effects.}
We consider the constraints
\begin{align}
x_t\in \mathcal{X},~~ u_t \in \mathcal{U}, \label{eq:problem_formulation:constraints}
\end{align}
where $\mathcal{X} \subseteq \R^{n_x}$ and $\mathcal{U} \subseteq \R^{n_u}$ are known convex sets.
The input $u_t=\operatorname*{MPC}(x_t,\theta)$ depends on the state $x_t$ (last constraint in \eqref{eq:problem_formulation:mpc}) and the parameters $\theta$, and it is computed by solving
\begin{align}
\begin{split}
\operatorname*{minimize}_{x_{\cdot|t},u_{\cdot|t}} & \quad \ell_{N,\theta}(x_{N|t}) + \sum_{k=0}^{N-1} \ell_\theta(x_{k|t},u_{k|t})\\
\text{subject to} & \quad x_{k+1|t} = g(x_{k|t},u_{k|t}),~~ x_{0|t} = x_t,~~ k\in\Z_{[0,N-1]},\\
& \quad x_{k|t} \in \mathcal{X},~~ k\in\Z_{[0,N]}, \\
& \quad u_{k|t} \in \mathcal{U},~~ k\in\Z_{[0,N-1]},
\end{split} \label{eq:problem_formulation:mpc_problem}
\end{align}
where $\ell_\theta$ and $\ell_{N,\theta}$ are parameterized cost functions, and setting $u_t = u_{0|t}$. 
Throughout, we assume that \eqref{eq:problem_formulation:mpc_problem} is a quadratic program (QP) meeting the conditions of Theorem~\ref{theorem:diff} in Appendix\,\ref{section:preliminaries:differentiating_optimization_problems}. 
If $g$ is a nonlinear function, one can use the linearization techniques in Section VI-A of \cite{zuliani2023bp}, to obtain an MPC that can be expressed as a quadratic program. 

Our emphasis is on QP-based MPC policies, which often deliver strong performance even on nonlinear control problems. 
Nonetheless, the framework can easily be extended to fully nonlinear policies using the differentiation methods in \cite{zuliani2025differentiable}. 
Similarly, this method can accommodate nonconvex upper-level constraints in \eqref{eq:problem_formulation:constraints} as long as the MPC satisfies Assumption~\ref{assumption:algorithm:definability_dynamics_and_cost}. 
Note additionally that while in this work we restrict $\theta$ to parameters appearing in the cost of \eqref{eq:problem_formulation:mpc}, extending it to also affect the constraints or dynamics of the MPC is straightforward.

Our goal is to obtain an MPC design $\theta^*$ that minimizes a known objective function $C(x,u)$, 
where $x=(x_0,\dots,x_T)$ and $u=(u_0,\dots,u_{T-1})$, typically with $T\gg N$, are the closed-loop trajectories obtained by combining \eqref{eq:problem_formulation:system_true} and \eqref{eq:problem_formulation:mpc}, 
while satisfying \eqref{eq:problem_formulation:constraints} for all $t$.
\begin{align}
\begin{split}
\operatorname*{minimize}_\theta & \quad C(x,u) = C(x_0,\dots,x_T,u_0,\dots,u_{T-1})\\
\text{subject to} & \quad x_{t+1} = f(x_t,u_t),~~x_0 \text{ given}, ~~ t\in\Z_{[0,T-1]}, \\
& \quad u_t = \operatorname*{MPC}(x_t,\theta),~~ t\in\Z_{[0,T-1]}, \\
& \quad x_t \in \mathcal{X},~~ t\in\Z_{[0,T]},\\
& \quad u_t \in \mathcal{U},~~ t\in\Z_{[0,T-1]}.
\end{split} \label{eq:problem_formulation:policy_optimization_problem}
\end{align}
Since $f$ is unknown, problem \eqref{eq:problem_formulation:policy_optimization_problem} cannot be solved directly. 
\rz{Our approach combines model-based gradient estimation, leveraging the available model $g$, with zeroth-order techniques. }

%% file: Chapters/model_based.tex
\subsection{Solving the problem when dynamics are known}

Even with perfect model knowledge, 
\eqref{eq:problem_formulation:policy_optimization_problem} is difficult to solve due to the nonsmooth constraints imposed by the MPC function.
\cite{zuliani2023bp} introduced a gradient-based method specifically designed for such problems, which solves an unconstrained reformulation of \eqref{eq:problem_formulation:policy_optimization_problem}.
To cast \eqref{eq:problem_formulation:policy_optimization_problem} as an unconstrained problem, 
let $x:\Theta\to \R^{(T+1)n_x}$ and $u:\Theta\to\R^{Tn_u}$ be the closed-loop trajectories obtained by combining \eqref{eq:problem_formulation:mpc} and \eqref{eq:problem_formulation:system_true} over the entire horizon $t\in\Z_{[0,T]}$.
Then \eqref{eq:problem_formulation:policy_optimization_problem} becomes
\begin{align}
\begin{split}
\operatorname*{minimize}_{\theta\in\Theta} & \quad C(x(\theta),u(\theta)) \\
\text{subject to} & \quad x(\theta) \in \mathcal{X}^{T+1},~~ u(\theta)\in \mathcal{U}^{T}.
\end{split} \label{eq:model_based_estimation:unconstrained_policy_optimization_problem_temp}
\end{align}
Since \eqref{eq:problem_formulation:mpc_problem} enforces $u_t(\theta)\in \mathcal{U}$ for all $t\in\Z_{[0,T-1]}$, 
the input constraints in \eqref{eq:model_based_estimation:unconstrained_policy_optimization_problem_temp} can be dropped.
The state constraints can be incorporated in the cost through a continuous penalty function $P(\theta)=\bar{P}(x(\theta))$
satisfying $P(\theta) = 0$ whenever $x(\theta)\in \mathcal{X}^{T+1}$ and $P(\theta) > 0$ otherwise
\begin{align}
\operatorname*{minimize}_{\theta\in\Theta} ~ C(x(\theta),u(\theta)) + P(\theta) =: \CC(\theta). \label{eq:model_based_estimation:unconstrained_policy_optimization_problem}
\end{align}
Under appropriate assumptions, which we detail in Section~\ref{section:convergence}, $\CC$ is path-differentiable and the update law $\theta_{k+1} = \Pi_\Theta [\theta_k - \alpha_k d_k]$, 
where $d_k \in \JC(\theta_k)$, and $\Pi_\Theta:\R^{n_\theta}\to\Theta$ is the projector onto the set $\Theta$,
converges to a minimizer of \eqref{eq:model_based_estimation:unconstrained_policy_optimization_problem_temp}.

\subsection{Imperfect Gradient Information using an Inexact Model}\label{section:model_based_imperfect}

Since \eqref{eq:problem_formulation:system_true} is unknown, the exact Jacobian $\JC$ cannot be computed,
and the exact update cannot be applied directly.
Instead, we approximate $\JC$ using the available model $g$ of the true dynamics $f$.
First, we recursively build approximations $\jj_x(\theta)$ and $\jj_u(\theta)$ of $\J_x(\theta)$ and $\J_u(\theta)$ via
\begin{subequations}
\label{eq:model_based_estimation:backprop_approximate_1}
\begin{align}
\jj_{x_{t+1}}(\theta) & = \jj_{g,x}(x_t,u_t) \jj_{x_{t}}(\theta) + \jj_{g,u}(x_t,u_t) \jj_{u_t}(\theta),\\
\jj_{u_t}(\theta) & = \jj_{\operatorname*{MPC},x} (x_t,\theta) \jj_{x_t}(\theta) + \jj_{\operatorname*{MPC},\theta}(x_t,\theta),
\end{align}
where $x_t=x_t(\theta)$, $u_t=u_t(\theta)$, and $J_g(x,u)$ and $J_\text{MPC}(x,\theta)$ are elements of the conservative Jacobians $\J_g(x,u)$ and $\J_\text{MPC}(x,\theta)$, respectively.
We then combine $J_x(\theta)$ and $J_u(\theta)$ to obtain an estimate $\jc(\theta)$ of $\JC(\theta)$ using the chain rule
\begin{align}
\jc(\theta) = \jj_{C,x}(x,u) \jj_x(\theta) + \jj_{C,u}(x,u) \jj_u(\theta), ~~ \jj_{C}(x,u) \in \J_{C}(x,u). \label{eq:model_based_estimation:backprop_approximate_2}
\end{align}
\end{subequations}
Replacing $\JC(\theta)$ with $\jc(\theta)$ yields the implementable update $\theta_{k+1} = \Pi_\Theta [\theta_k - \alpha_k \jc(\theta_k)]$.

%% file: Chapters/model_free.tex
The update of Section~\ref{section:model_based_imperfect} is appealing because it only relies on the approximate model. 
However, without further assumptions on the model accuracy, it is impossible to derive convergence guarantees. 
To address this issue, we construct gradient-like directions purely from data, 
for which convergence can be established under mild conditions. 

For any $\delta>0$, consider the \emph{randomized smoothing approximation} $\CCd(\theta):\Theta_\delta\to \R$ of $\CC$
\begin{align*}
\CCd(\theta) = \mathbb{E}_{w \sim U(\mathbb{B})} [\CC(\theta+\delta w)],
\end{align*}
where $\Theta_\delta := \Theta + \delta \mathbb{B}$.
$\CCd$ is a smooth approximation of $\CC$: as long as $\CC$ is Lipschitz with constant $L_{\CC}$,
$\CCd$ is continuously differentiable and Lipschitz with constant $L_{\CCd}=c\sqrt{n_\theta}\delta^{-1}L_{\CC}$ for some $c>0$ \cite[Proposition 2.3]{lin2022gradient}.
Working with $\CCd$ instead of $\CC$ has two key benefits:
(i) the Lipschitz continuity of $\CC$ implies smoothness of $\CCd$, and (ii) one can estimate $\nabla\CCd$ solely using function evaluations.
For this, we use the following one-point estimator
\begin{align}\label{eq:model_free_estimation:zeroth_order}
 \jct(\theta_k,v_k) = \frac{n_\theta}{\delta} [\CC(\theta_k+\delta v_k)-\CC(\theta_{k})] v_k,
\end{align}
where $v_k$ is sampled i.i.d.\ from $U(\mathbb{S})$ for each $k\in\N$.
We have the following.
\begin{lemma}\label{lemma:model_free_estimation:gradient_sampling}
For any $\theta\in\Theta$, $\mathbb{E}_v[\jct(\theta,v)]=\nabla \CCd(\theta)$.
\end{lemma}
\begin{proof}
By Lemma 1 in \cite{flaxman2004online}, it holds that $\mathbb{E}_v[\mathcal{C}(\theta+\delta v)v] = \delta/n_\theta \nabla \mathcal{C}^\delta(\theta)$. Since $v$ is zero-mean, we immediately have $\mathbb{E}_v[(\mathcal{C}(\theta+\delta v)-\mathcal{C}(\theta))v]=\mathbb{E}_v[\mathcal{C}(\theta+\delta v)v] = \delta/n_\theta \nabla \mathcal{C}^\delta(\theta)$.
\end{proof}
Lemma~\ref{lemma:model_free_estimation:gradient_sampling} shows that the zeroth-order estimator \eqref{eq:model_free_estimation:zeroth_order} provides an unbiased estimate of $\nabla \CCd$. 
To relate this to the original nonsmooth objective, we recall the following result from \cite{lin2022gradient}.
\begin{lemma}\label{lemma:model_free_estimation:delta_subgradient}
Suppose $\CC$ is $L_{\CC}$-Lipschitz in $\Theta$. Then $\nabla \CCd(\theta) \in \partial_\delta \CC(\theta)$ for any $\theta\in\Theta$.
\end{lemma}
This means that estimating $\nabla\CCd$ yields an element of the Goldstein $\delta$-subdifferential of the true objective $\CC$,
thus providing a second implementable update law $\theta_{k+1} = \Pi_\Theta [\theta_k - \alpha_k \jct(\theta_k,v_k)]$.

%% file: Chapters/proposed_algorith.tex
Our algorithm combines the model-based update of Section~\ref{section:model_based_imperfect} with the model-free update of Section~\ref{section:model_free}, following the \emph{gray-box} scheme of \cite{he2024gray}.
At iteration $k$, we update $\theta^k$ as follows
\begin{align}
\begin{split}
d_{k,1} & = \jc(\theta_k),\\
d_{k,2} & = \jct(\theta_k,v_k),~ v_k \sim U(\mathbb{S}),\\
d_k & = \eta_k d_{k,1} + (1-\eta_k) d_{k,2},\\
\theta_{k+1} & = \Pi_\Theta[\theta_k - \alpha_k d_k],
\end{split} \label{eq:algorithm:main_update}
\end{align}
where $\{ \alpha_k \}_{k\in\N} \subset \R_{>0}$ is a sequence of vanishing stepsizes,
and $\{ \eta_k \}_{k\in\N} \subset [0,1]$ weights the relative contribution of the two update directions.
Choosing $\eta_k \approx 1$ prioritizes the model-based direction $\jc(\theta_k)$, 
whereas $\eta_k \approx 0$ makes the update closer to zeroth order via $\jct(\theta_k,v_k)$.
Generally, selecting $\eta_k$ is a design choice that should reflect the trustworthiness of the model. We showcase how the converge speed is affected by different choices of $\eta_k$ in simulation in Section\,\ref{section:simulation}.

\subsection{Convergence to a Goldstein \texorpdfstring{$\delta$}{delta}-Critical Point}\label{section:convergence}

Our convergence analysis builds on Lemma~\ref{lemma:model_free_estimation:gradient_sampling} and standard results on stochastic projected gradient methods for nonsmooth, nonconvex objectives \citep{davis2020stochastic}.
We first impose regularity assumptions ensuring that the closed-loop map and the objective are locally Lipschitz and definable.
\begin{assumption}\label{assumption:algorithm:definability_dynamics_and_cost}
The true dynamics $f$, the model $g$, the cost function $C$, the penalty $P$, and the MPC function $\operatorname*{MPC}:\R^{n_x}\times\Theta\to\R^{n_u}$ are locally Lipschitz and definable in an o-minimal structure.
\end{assumption}
Under Assumption~\ref{assumption:algorithm:definability_dynamics_and_cost}, the closed-loop trajectories $x(\theta)$ and $u(\theta)$ are locally Lipschitz and definable in $\theta$,
allowing the definition of conservative Jacobians.
Assumption~\ref{assumption:algorithm:definability_dynamics_and_cost} is not restrictive, as definable functions cover almost all functions of interest in control and optimization,
and the MPC function satisfies Assumption~\ref{assumption:algorithm:definability_dynamics_and_cost} under the mild conditions laid out in Appendix~\ref{section:preliminaries:differentiating_optimization_problems}.
To ensure feasibility of the MPC, one can resort to the technique in Section VI-D of \cite{zuliani2023bp}.

Our last technical requirement simplifies the analysis by ensuring boundedness of the gradients.
\begin{assumption}\label{assumption:algorithm:parameter_set_is_compact_and_definable}
The set $\Theta$ is convex, compact and definable in an o-minimal structure.
\end{assumption}
To ensure convergence, we require the following conditions on $\alpha_k$ and $\eta_k$
\begin{subequations}
\label{eq:algorithm:stepsize_condition}\begin{align}
& \alpha_k > 0, ~~ && {\textstyle\sum_{k\in\N}} \alpha_k = + \infty,~~ {\textstyle\sum_{k\in\N}} \alpha_k^2 < + \infty, \label{eq:algorithm:stepsize_condition:alpha}\\
& \eta_k \in [0,1],~~ && {\textstyle\sum_{k\in\N}} \eta_k \alpha_k < + \infty. \label{eq:algorithm:stepsize_condition:eta}
\end{align}
\end{subequations}
These conditions hold, for example, if $\alpha_k = 1/(k+1)^\gamma$ and $\eta_k = 1/(k+1)^{\beta-\gamma}$, with $\gamma\in(0.5,1]$ and $\beta > 1$.
This allows for a wide range of possible stepsizes and,
crucially, for different decrease rates for $\eta_k$.
This last feature, in particular,
allows us to accomodate situations where the model $\jc$ is deemed trustworthy,
and thus $\eta_k$ should have larger values,
but also situations where $\jc$ is less trusted and the zeroth-order estimation is preferred.
Observe that $\eta_k \downarrow 0$, meaning that eventually the information obtained using the model is discarded and the algorithm relies solely on data.

\begin{theorem}\label{theorem:algorithm:main_convergence_result}
Under Assumptions~\ref{assumption:algorithm:definability_dynamics_and_cost}, and \ref{assumption:algorithm:parameter_set_is_compact_and_definable}, if $\alpha_k$ and $\eta_k$ satisfy \eqref{eq:algorithm:stepsize_condition}, then $\theta_k$ as obtained through \eqref{eq:algorithm:main_update} converges to a Goldstein Fritz-John $\delta$-critical point of the problem
\begin{align}
\operatorname*{minimize}_{\theta\in \Theta} \quad \mathcal{C}(\theta) \quad \text{subject to}~~ P(\theta)=0. \label{eq:algorithm:problem_penalty}
\end{align}
\end{theorem}
The proof of Theorem~\ref{theorem:algorithm:main_convergence_result} requires several preliminary results.
First, we prove that under Assumption~\ref{assumption:algorithm:definability_dynamics_and_cost}
the approximation $\CCd$ of $\CC$ retains definability and Lipschitz continuity.
\begin{lemma}\label{lemma:algorithm:randomized_smoothing_is_definable}
Under Assumptions~\ref{assumption:algorithm:definability_dynamics_and_cost} and \ref{assumption:algorithm:parameter_set_is_compact_and_definable}, $\CCd$ is Lipschitz continuous and definable.
\end{lemma}
\begin{proof}
Under Assumptions~\ref{assumption:algorithm:definability_dynamics_and_cost} and \ref{assumption:algorithm:parameter_set_is_compact_and_definable},
the functions $x(\theta)$ and $u(\theta)$ are definable and locally Lipschitz since both these properties are preserved by composition \cite[Exercise 1.11]{coste1999introduction}.
Next, the function $y\mapsto \max \{ y_1,y_2 \}$ is the pointwise maximum of two linear functions, 
and it is therefore locally Lipschitz and definable 
(it is, in fact, semialgebraic).
This proves that $\CC$ is locally Lipschitz and definable. 
Since integration (and therefore expectation) preserves definability \citep{speissegger1999pfaffian}, $\CCd$ is definable for every $\delta>0$.
Moreover, since $\CC$ is locally Lipschitz and therefore Lipschitz if restricted to $\Theta$,
by \cite[Proposition 2.3]{lin2022gradient},
$\CCd$ is Lipschitz for any $\delta>0$.
\end{proof}
Next, we show that the zeroth-order update dominates the model-based one for all $k$ large enough.
\begin{lemma}\label{lemma:algorithm:data_drive_eventually_dominates}
Under Assumptions~\ref{assumption:algorithm:definability_dynamics_and_cost} and \ref{assumption:algorithm:parameter_set_is_compact_and_definable},
if $\alpha_k$ and $\eta_k$ satisfy \eqref{eq:algorithm:stepsize_condition},
then $\sum_{k\in\N} \alpha_k \eta_k \|\jc(\theta_k)\|< + \infty$.
\end{lemma}
\begin{proof}
Since $\sum_{k\in\N} \alpha_k \eta_k < + \infty$ by \eqref{eq:algorithm:stepsize_condition:eta},
it suffices to prove that $\|\jc(\theta_k)\|$ is bounded for all $\theta_k$.
To prove this,
observe that $\|x(\theta)\| \leq C_x$ and $\|u(\theta)\| \leq C_u$ for all $\theta\in\Theta$ for some (unknown) $C_x,C_u < + \infty$ since both $x(\theta)$ and $u(\theta)$ are locally Lipschitz and $\Theta$ is compact by Assumption~\ref{assumption:algorithm:parameter_set_is_compact_and_definable}. 
Since $J_{g}$ in \eqref{eq:model_based_estimation:backprop_approximate_1} is an element of the conservative Jacobian of the locally Lipschitz definable function $g$,
its value is almost surely equal to $\nabla g$,
and it is therefore almost surely bounded above by the Lipschitz constant of $g$ on $\Theta$.
The same goes for $J_{\operatorname*{MPC}}$.
Since $J_x$ and $J_u$ are constructed through the recursion \eqref{eq:model_based_estimation:backprop_approximate_1} involving bounded quantities,
and the horizon $T$ of the problem is finite,
$J_{x_t}$ and $J_{u_t}$ are bounded for all $t$,
and therefore so are $J_x$ and $J_u$.
Finally, since $C$ is Lipschitz on the set $C_x \mathbb{B} \times C_u \mathbb{B}$,
$\jc$ is bounded.
This completes the proof.
\end{proof}
Our final technical result is about proving the finiteness of the variance of $\jct$ for each $k$.
\begin{lemma}\label{lemma:algorithm:jct_has_finite_variance}
Under Assumptions~\ref{assumption:algorithm:definability_dynamics_and_cost} and \ref{assumption:algorithm:parameter_set_is_compact_and_definable},
we have for all $k\in\N$ that
\begin{align*}
\mathbb{E}_{v_k}[\jct(\theta_k,v_k)-\nabla \CCd(\theta_k)]=0, ~~ \mathbb{E}_{v_k}[\|\jct(\theta_k,v_k)-\nabla \CCd(\theta_k)\|^2] < + \infty.
\end{align*}
\end{lemma}
\begin{proof}
The first equation is trivial since $\mathbb{E}[\jct(\theta_k,v_k)]=\nabla \CCd(\theta_k)$ by Lemma~\ref{lemma:model_free_estimation:gradient_sampling}.
Next,
we have
\begin{align*}
\mathbb{E}_{v_k}[\|\jct(\theta_k,v_k)-\nabla \CCd(\theta_k)\|^2] & = - \|\nabla \CCd(\theta_k)\|^2 \! + \mathbb{E}_{v_k}[\|\jct(\theta_k,v_k)\|^2] \\
& = - \|\nabla \CCd(\theta_k)\|^2 \! + \mathbb{E}_{v_k}\left[\frac{n_\theta^2 \|v_k\|^2}{\delta^2}(\CC(\theta_k+\delta v_k)-\CC(\theta_{k}))^2\right]\\
& \leq - \|\nabla \CCd(\theta_k)\|^2 \! + \frac{n_\theta^2 L_{\CC}}{\delta^2} \mathbb{E}_{v_k}[\|v_k\|^2 \|\theta_k+\delta v_k - \theta_{k} \|^2], \\
& \leq - \|\nabla \CCd(\theta_k)\|^2 \! + n_\theta^2 L_{\CC} \mathbb{E}_{v_k}[\|v_k\|^4],
\end{align*}
% where the first step follows from $\mathbb{E}_{v_k}[\jct(\theta_k,v_k)]=\nabla \CCd(\theta_k)$. 
Since $v_k \sim U(\mathbb{S})$, the term on the right is always finite for finite $\theta_k$.
Combining this with the continuity of $\|\nabla \CCd(\cdot)\|$,
we conclude the existence of a function $p:\Theta\to\R_{>0}$ bounded on bounded sets such that $\mathbb{E}_{v_k}[\|\jct(\theta_k)-\nabla \CCd(\theta_k)\|^2] \leq p(\theta_k)$,
concluding the proof.
\end{proof}

\noindent\textbf{Proof of Theorem~\ref{theorem:algorithm:main_convergence_result}}.
We follow Section A in \cite{davis2020stochastic}. 
Note that since $\mathcal{C}^\delta$ is continuously differentiable, we can take its gradient $\nabla \mathcal{C}^\delta$ as a conservative field. 
First, observe that
\begin{align} \label{eq:algorithm:equivalent_update_1}
\theta_{k+1} = \Pi_\Theta [ \theta_k - \alpha_k \nabla \CCd(\theta_k) ] + \alpha_k \xi_k
\end{align}
where $\alpha_k\xi_k = \Pi_\Theta[\theta_k-\alpha_k(1-\eta_k)\jct(\theta_k,v_k)-\alpha_k \eta_k \jc(\theta_k)]- \Pi_\Theta [\theta_k - \alpha_k \nabla \CCd(\theta_k) ]$.
By leveraging the convexity of $\Theta$ and the triangle inequality,
we have $\|\xi_k\| \leq \eta_k [ \|\jc(\theta_k)\| + \| \jct(\theta_k,v_k) \| ] + \|\jct(\theta_k,v_k)-\nabla \CCd(\theta_k)\|$.
% \begin{align*}
% \|\xi_k\| \leq \eta_k [ \|\jc(\theta_k)\| + \| \jct(\theta_k,v_k) \| ] + \|\jct(\theta_k,v_k)-\nabla \CCd(\theta_k)\|.
% \end{align*}
Since $\CCd$ is Lipschitz and definable by Lemma~\ref{lemma:algorithm:randomized_smoothing_is_definable},
and both $\theta_k\in \Theta$ and $v_k\in\mathbb{S}$ take on finite values,
there exists a constant $C_J>0$ such that $\|\jc(\theta_k)\|+\|\jct(\theta_k,v_k)\| \leq C_J$ for all $k\in\N$,
which combined with \eqref{eq:algorithm:stepsize_condition:eta} gives $\sum_{k\in\N} \alpha_k \eta_k [ \|\jc(\theta_k)\| + \| \jct(\theta_k,v_k) \| ] < + \infty$.
Next, leveraging Lemma~\ref{lemma:algorithm:data_drive_eventually_dominates}, \eqref{eq:algorithm:stepsize_condition:alpha},
and the compactness of $\Theta$,
we conclude that $\sum_{k\in\N} \alpha_k \|\jct(\theta_k)-\nabla \CCd(\theta_k)\| < + \infty$ by \cite[Lemma 4.1]{davis2020stochastic}, as the summability condition coincides with Assumption A.4 in \cite{davis2020stochastic}.
This proves that ${\textstyle \sum_{k\in\N}} \alpha_k \xi_k < + \infty$. 
Next, letting $G_k(\theta) = - \nabla \CCd(\theta_k) - \alpha_k^{-1}[\theta-\alpha_k \nabla \CCd(\theta)-\Pi_\Theta[\theta-\alpha_k \nabla \CCd(\theta)]]$, 
the update in \eqref{eq:algorithm:equivalent_update_1} can be written as 
% $\theta_{k+1} = \theta_k + \alpha_k[g_k+\xi_k],~~~ g_k \in G_k(\theta_k)$.
\begin{align*}
\theta_{k+1} = \theta_k + \alpha_k[g_k+\xi_k],~~~ g_k \in G_k(\theta_k).
\end{align*}
The final argument of this proof relies on Theorem 3.2 of \cite{davis2020stochastic},
which we will invoke to prove convergence to a critical point of \eqref{eq:model_based_estimation:unconstrained_policy_optimization_problem_temp}.
To utilize Theorem 3.2 of \cite{davis2020stochastic} we require all items in Assumption A of \cite{davis2020stochastic} to hold true.
First, observe that items 1-4 hold thanks to Assumption~\ref{assumption:algorithm:parameter_set_is_compact_and_definable} and $\sum_{k\in\N} \alpha_k \xi_k < + \infty$. 
It only remains to show that item 5 holds, that is, 
that given any unbounded subset $\mathcal{K}$ of $\N$ for which $\theta_j \to \bar{\theta}$, $j\in\mathcal{K}$,
we have $\operatorname*{dist} ( 1/k {\textstyle\sum_{j=0}^{k}} g_j, G(\bar{\theta}) ) \to 0$, where $G(\bar{\theta}) = - \nabla \CCdinf(\bar{\theta}) - \mathcal{N}_\Theta(\bar{\theta})$ and $\CCdinf(\theta)=\mathbb{E}_v[\mathcal{C}(\theta+\delta v)]$.
Since $G(\bar{\theta})$ is a convex set, we have $\operatorname*{dist} ( 1/k {\textstyle\sum_{j=0}^{k}} g_j, G(\bar{\theta}) ) \leq 
1/k {\textstyle\sum_{j=0}^{k}} \operatorname*{dist} (g_j,G(\bar{\theta})),$
% \begin{align*}
% \operatorname*{dist} \left( 1/k {\textstyle\sum_{j=0}^{k}} g_j, G(\bar{\theta}) \right) \leq 
% 1/k {\textstyle\sum_{j=0}^{k}} \operatorname*{dist} (g_j,G(\bar{\theta})),
% \end{align*}
meaning that it suffices to show that $\operatorname*{dist} (g_j,G(\bar{\theta}))\to 0$ as $j\to \infty$, $j\in \mathcal{K}$.
Since for each $j$ we have $\Pi_\Theta[\theta_j-\alpha_j \nabla \CCd(\theta_j)] \in \theta_j - \alpha_j \nabla \CCd(\theta_j) - \mathcal{N}_\Theta(\Pi_\Theta[\theta_j-\alpha_j \nabla \CCd(\theta_j)])$,
we have by definition that $g_j = -\nabla \CCd(\theta_j) - \alpha_j^{-1} z_j$, for some $z_j \in \mathcal{N}_\Theta(\Pi_\Theta[\theta_j-\alpha_j \nabla \CCd(\theta_j)])$, and therefore $g_j - G(\bar{\theta}) = -\nabla \CCd(\theta_j) - \alpha_j^{-1} z_j + \nabla \CCdinf(\bar{\theta}) + \mathcal{N}_\Theta(\bar{\theta})$. 
% \begin{align*}
% g_j - G(\bar{\theta}) = -\nabla \CCd(\theta_j) - \alpha_j^{-1} z_j + \nabla \CCdinf(\bar{\theta}) - \mathcal{N}_\Theta(\bar{\theta}).
% \end{align*}
Due to the outer semicontinuity of the normal cone $\mathcal{N}_\Theta$ of a convex set $\Theta$, 
and that $\theta_{j+1}=\Pi_\Theta[\theta_j-\alpha_j \nabla \CCd(\theta_j)] \to \bar{\theta}$,
we have that in the limit $\alpha_j^{-1} z_j \in \mathcal{N}_\Theta(\bar{\theta})$.
Moreover, by continuity of $\nabla\CCd$, $\nabla \CCd(\theta_j) \to \nabla \CCdinf(\bar{\theta})$.
This proves that all items in Assumption A of \cite{davis2020stochastic} are satisfied.
Since Assumption B of \cite{davis2020stochastic} is also satisfied thanks to Lemma~\ref{lemma:algorithm:randomized_smoothing_is_definable} and Theorem 5.8 of \cite{davis2020stochastic},
we conclude that \eqref{eq:algorithm:main_update} converges to a point satisfying $0\in \nabla \CCd(\bar{\theta}) + \mathcal{N}_\Theta(\bar{\theta})$. 
By Lemma~\ref{lemma:model_free_estimation:delta_subgradient},
this means that $0\in \partial_\delta [C(\bar{\theta})+P(\bar{\theta})]+ \mathcal{N}_\Theta(\bar{\theta})$.
Since $\bigcup_{\vartheta\in \delta\mathbb{B}} [\partial_c C(\theta+\vartheta) + \partial_c P(\theta+\vartheta)] \subseteq \bigcup_{\vartheta\in \delta\mathbb{B}} [\partial_c C(\theta+\vartheta)] + \bigcup_{\vartheta\in \delta\mathbb{B}}[\partial_c P(\theta+\vartheta)]$, and this inclusion is preserved if we consider the convex hull of both sets, we have that $\partial_\delta [C(\bar{\theta})+P(\bar{\theta})] \subseteq \partial_\delta C(\bar{\theta}) + \partial_\delta P(\bar{\theta})$, and therefore there exist $\lambda_0,\lambda_1 \geq 0$ with $\lambda_0+\lambda_1=1$ such that
\begin{align}
0 \in \lambda_0 \partial_\delta C(\bar{\theta}) + \lambda_1 \partial_\delta P(\bar{\theta}) + \mathcal{N}_\Theta(\bar{\theta}), ~~ \lambda_0,\lambda_1 \geq 0, ~~ \lambda_0+\lambda_1=1, \label{eq:algorithm:fritz_john_condition}
\end{align}
proving that the algorithm converges to a Goldstein Fritz-John $\delta$-critical point of \eqref{eq:algorithm:problem_penalty}.\hfill\scalebox{1.3}{$\blacksquare$}\\[-0.5em]

\noindent Under stronger assumptions on the constraint $P(\theta) = 0$ (for example the constraint qualification given in \cite{grimmer2025goldstein}),
one can additionally prove that any feasible solution $\bar{\theta}$ satisfies \eqref{eq:algorithm:fritz_john_condition} with $\lambda_0>0$,
that is, that $\bar{\theta}$ is a KKT point of \eqref{eq:model_based_estimation:unconstrained_policy_optimization_problem}.
\rz{We leave this as a direction for future work.}

%% file: Chapters/simulation_results.tex
We evaluate our approach on the 12-dimensional quadcopter model from \cite{abdulkareem2022modeling} whose state vector comprises the position $(p_x,p_y,p_z)$, velocity $(v_x,v_y,v_z)$, Euler angles $(\phi,\vartheta,\psi)$, and angular velocity $(p,q,r)$ in the body frame.\footnote{The code will be made available at \url{https://github.com/RiccardoZuliani98/bpmpc}}
The control inputs are the rotation speeds $\omega_i$ of the four rotors.
The system is subject to the constraints $\omega_i\in [0,630]$, $v_x,v_y,v_z\in[-2,2]$, $\phi,\vartheta$, $\psi\in [-\pi/4,\pi/4]$, and $p,q,r\in[-\pi/8,\pi/8]$. 
We implement \eqref{eq:problem_formulation:mpc_problem} with $N=12$ and
\begin{align*}
\ell_\theta(x,u) = \|x-x_\text{ref}\|_Q^2 + \|u-u_\text{ref}\|_R^2, \quad \ell_{N,\theta}(x) = \|x-x_\text{ref}\|_P^2,
\end{align*}
where $x_\text{ref}=(-6,-3.5,0,\mathbf{0}_9)$ and $u_\text{ref}$ is the input required to maintain the drone at a hovering state. The parameter $\theta=(p_Q,p_R,p_P)$ defines the stage cost matrices $Q=\operatorname*{diag}(p_Q^2)+10^{-6}I$ and $R=\operatorname*{diag}(p_R^2)+10^{-6}I$, as well as the terminal cost $P=LL^\top+10^{-6}I$, where $L$ is a lower-triangular matrix containing the entries of $p_P$.
To handle constraint violations, we relax the state constraints using slack variables, which are penalized using both a quadratic and a linear penalty (scaled by a factor of $25$).
We model the system as linear choosing \( g(x,u) = A x + B u \), where the matrices \( A \) and \( B \) are identified via least-squares regression on 100 closed-loop trajectories collected near the target point under a stabilizing MPC controller. In practice, if such a controller is not available, these trajectories could instead be generated by a human pilot.
The upper-level cost is
\begin{align*}
C(x,u) &= \|x_T-x_\text{ref}\|_\mathcal{P}^2 + {\textstyle\sum_{t=0}^{T-1}} \|x_t-x_\text{ref}\|_\mathcal{Q}^2 + \|u_t-u_\text{ref}\|_\mathcal{R}^2
\end{align*}
where $T=200$, $\mathcal{Q}=\operatorname*{diag}(\mathbf{1}_6,0.1\cdot\mathbf{1}_6)$, $\mathcal{R}=0.01\cdot I$,
and $\mathcal{P}$ obtained by solving the DARE with the identified dynamics and cost given by $\mathcal{Q}$ and $\mathcal{R}$.
The penalty term is $P(\theta)=10^4\operatorname*{dist}_1(x(\theta),\mathcal{X}^{T+1})$.

We choose $\delta=10^{-5}$, $\alpha_k=2 \cdot 10^{-4} \log(k+2) / (k+1) ^ {0.8}$ and $\eta_k = 1/(k+1)^\gamma$, with $\gamma\in \{ 0.25,0.5,0.75 \}$, and train for $300$ iterations starting with $Q=\mathcal{Q}$, $R=\mathcal{R}$, and $P=\mathcal{P}$.
To isolate the contribution of each optimization component, we additionally train using the same stepsizes and initial conditions with $\eta_k=1$ and $\eta_k=0$.
The behavior of the tracking cost and the constraint violation across iterations can be seen in Figure~\ref{fig:tracking_plots}. The proposed approach demonstrates the most efficient convergence, whereas the purely model-based method exhibits rapid improvement during the initial iterations but later becomes unstable and fails to converge. The purely zeroth-order method, on the other hand, achieves consistent progress, but at a slower rate. The proposed hybrid strategy benefits from fast initial transients, enabled by the model, and improved long-term convergence, supported by zeroth-order correction. 
Notably, smaller values of $\gamma$, corresponding to greater reliance on the model, lead to faster initial convergence, but also to more pronounced instability.

\begin{figure}[t]
    \centering
    % OPTION 1: tex
    \ifcompile
      \pgfplotsset{table/search path={Figures/Quadcopter}}
      \input{Figures/Quadcopter/quadcopter_figure.tex}~~
    % OPTION 2: pdf
    \else
      \includegraphics{Figures/Quadcopter/quadcopter.pdf}
    \fi
    \vspace{-0.2cm}
    \caption{Tracking cost and constraint violation across iterations.}
    \label{fig:tracking_plots}
    \vspace{-0.2cm}
\end{figure}

To further assess policy quality, we compare the trajectory generated by the trained controller with that obtained from tuning using the exact nonlinear model.  
As shown in Figure~\ref{fig:trajectory_plots}, the two trajectories exhibit strong qualitative agreement across all states. 
Quantitatively, the trained controller achieves a final cost of $1084.91$, compared to $1081.3$ for the controller optimized with perfect model knowledge. This corresponds to less than $1\%$ suboptimality, indicating that the learned parameters are nearly optimal despite using an approximate model and noisy gradients.

More simulation results are reported in Appendix~\ref{section:ablation}.

\begin{figure}[t]
    \centering
    % OPTION 1: tex
    \ifcompile
      \pgfplotsset{table/search path={Figures/Quadcopter_trajectory}}
      \input{Figures/Quadcopter_trajectory/quadcopter_trajectory_figure.tex}
      \vspace{-0.6cm}
    % OPTION 2: pdf
    \else
      \includegraphics{Figures/Quadcopter_trajectory/quadcopter_trajectory.pdf}
      \vspace{-0.6cm}
    \fi
    \vspace{-0.2cm}
    \caption{Comparison of position (left) and attitude (right) trajectories obtained with the trained controller (solid) and the controller tuned using the exact model (marked dash-dotted).}
    \label{fig:trajectory_plots}
    \vspace{-0.7cm}
\end{figure}

%% file: Figures/Quadcopter/quadcopter_figure.tex
\begin{tikzpicture}
\begin{groupplot}[group style={group size=2 by 1, horizontal sep=10pt}]

% --- LEFT (MAIN) AXIS
\nextgroupplot[
  legend cell align={left},
  legend style={
      fill opacity=0.8, 
      draw opacity=1, 
      text opacity=1, 
      draw=grid_gray,
      row sep=-1.5pt,
      font=\small,
  },
  tick label style={font=\scriptsize},
  tick align=outside,
  tick pos=left,
  x grid style={grid_gray},
  xlabel={\small Iteration},
  xmajorgrids,
  xmin=-5, xmax=305,
  xtick style={color=black},
  y grid style={grid_gray},
  ylabel={\small Tracking cost},
  ymajorgrids,
  ymin=1000, ymax=2000,
  ytick style={color=black},
  width=7.5cm, height=4cm,
  line width=0.6,
  grid style={dashed,grid_gray,thin},
  scaled y ticks=base 10:-2
]
\addplot [smooth,model] table [col sep=space] {./cost/cost_mb.dat};
\addplot [zo] table [col sep=space] {./cost/cost_zo.dat};
\addplot [bl25] table [col sep=space] {./cost/cost_bl25.dat};
\addplot [bl50] table [col sep=space] {./cost/cost_bl5.dat};
\addplot [bl75] table [col sep=space] {./cost/cost_bl75.dat};

% --- ZOOM RECTANGLE (store corners as coordinates)
\coordinate (zoomNW) at (axis cs:5,1300);
\coordinate (zoomSE) at (axis cs:50,1100);
\draw[thick, connect_gray] (zoomNW) rectangle (zoomSE);

% --- RIGHT AXIS
\nextgroupplot[
  legend cell align={left},
  legend style={
      fill opacity=0.8, 
      draw opacity=1, 
      text opacity=1, 
      draw=grid_gray,
      row sep=-2.5pt,
      font=\footnotesize,
      legend columns=2,
  },
  axis y line=box,
  tick align=outside,
  x grid style={grid_gray},
  xlabel={\small Iteration},
  xmajorgrids,
  xmin=-9.95, xmax=305,
  xtick pos=left,
  xtick style={color=black},
  y grid style={grid_gray},
  ylabel={\small Constraint violation},
  ymajorgrids,
  ymin=-0.058038046702463, ymax=1.2,
  ytick pos=right,
  ytick style={color=black},
  yticklabel style={anchor=west},
  width=7.5cm, height=4cm,
  tick label style={font=\scriptsize},
  line width=0.6,
  grid style={dashed,grid_gray,thin}
]
\addplot [smooth, model] table [col sep=space] {cst/cst_mb.dat}; 
\addlegendentry{Model}
\addplot [bl25] table [col sep=space] {cst/cst_bl25.dat}; 
\addlegendentry{$\gamma=0.25$}
\addplot [bl50] table [col sep=space] {cst/cst_bl5.dat};  
\addlegendentry{$\gamma=0.5$}
\addplot [bl75] table [col sep=space] {cst/cst_bl75.dat}; 
\addlegendentry{$\gamma=0.75$}
\addplot [zo] table [col sep=space] {cst/cst_zo.dat}; 
\addlegendentry{Zeroth-order}
\end{groupplot}

% --- INSET AXIS (no ticks or labels)
\begin{axis}[
  at={($(group c1r1.north east)+(-2mm,-2mm)$)},
  anchor=north east,
  width=3.6cm, height=3cm,
  xmin=5, xmax=50,
  ymin=1100, ymax=1300,
  scaled y ticks=base 10:-2,
  axis line style={connect_gray, semithick},
  tick style={draw=none},
  xtick=\empty, ytick=\empty,
  xlabel={}, ylabel={},
  every axis plot/.append style={line width=0.35pt},
  axis background/.style={fill=white, fill opacity=1, draw=black},
]
\addplot [model] table [col sep=space] {./cost/cost_mb.dat};
\addplot [zo] table [col sep=space] {./cost/cost_zo.dat};
\addplot [bl25,smooth]    table [col sep=space] {./cost/cost_bl25.dat};
\addplot [bl50]   table [col sep=space] {./cost/cost_bl5.dat};
\addplot [bl75]  table [col sep=space] {./cost/cost_bl75.dat};

% save inset corners for connectors
\coordinate (insetNW) at (rel axis cs:0,1);
\coordinate (insetSE) at (rel axis cs:1,0);
\end{axis}

% --- CONNECTORS between zoom rectangle and inset
\draw[semithick, connect_gray] (zoomNW) -- (insetNW);
\draw[semithick, connect_gray] (zoomSE) -- (insetSE);

\end{tikzpicture}

%% file: Figures/Quadcopter_trajectory/quadcopter_trajectory_figure.tex
% This file was created with matplot2tikz v0.4.2.
\begin{tikzpicture}

\begin{groupplot}[group style={group size=2 by 1,horizontal sep=10pt}]
\nextgroupplot[
legend cell align={left},
legend style={
  fill opacity=0.8, 
  draw opacity=1, 
  text opacity=1, 
  draw=grid_gray,
  legend columns = 3,
  font=\small},
tick align=outside,
tick pos=left,
x grid style={grid_gray},
xlabel={\small Time [s]},
xmajorgrids,
xmin=-0.3, xmax=10.3,
xtick style={color=black},
width=7.5cm, height=4cm,
y grid style={grid_gray},
ylabel={\small Position [m]},
ymajorgrids,
ymin=-6.39922886793102, ymax=2.39996327942529,
ytick style={color=black},
line width=0.6,
grid style={dashed,grid_gray, semithick},
tick label style={font=\scriptsize},
]
\addplot [x_ours] table [col sep=space] {./trajectory/x_out_0.dat};
\addlegendentry{$x$}
\addplot [y_ours] table [col sep=space] {./trajectory/x_out_1.dat};
\addlegendentry{$y$}
\addplot [z_ours] table [col sep=space] {./trajectory/x_out_2.dat};
\addlegendentry{$z$}
\addplot [x_best, forget plot] table [col sep=space] {./trajectory/x_nlp_0.dat};
\addplot [y_best, forget plot] table [col sep=space] {./trajectory/x_nlp_1.dat};
\addplot [z_best, forget plot] table [col sep=space] {./trajectory/x_nlp_2.dat};

\nextgroupplot[
axis y line=box,
tick align=outside,
legend cell align={left},
legend style={
  fill opacity=0.8,
  draw opacity=1,
  text opacity=1,
  draw=grid_gray,
  legend columns = 3,
  font=\small},
tick align=outside,
x grid style={grid_gray,semithick},
xlabel={\small Time [s]},
xmajorgrids,
xmin=-0.3, xmax=10.3,
xtick pos=left,
xtick style={color=black},
y grid style={grid_gray,semithick},
ylabel={\small Euler Angles [rad]},
ymajorgrids,
ymin=-0.355775629679425, ymax=0.332098623712708,
ytick pos=right,
ytick style={color=black},
yticklabel style={anchor=west},
line width=0.6,
grid style={dashed,gray7,semithick},
width=7.5cm, height=4cm,
tick label style={font=\scriptsize},
]
\addplot [x_ours] table [col sep=space] {./trajectory/x_out_6.dat};
\addlegendentry{$\phi$}
\addplot [y_ours] table [col sep=space] {./trajectory/x_out_7.dat};
\addlegendentry{$\theta$}
\addplot [z_ours] table [col sep=space] {./trajectory/x_out_8.dat};
\addlegendentry{$\psi$}

\addplot [x_best] table [col sep=space] {./trajectory/x_nlp_6.dat};
\addplot [y_best] table [col sep=space] {./trajectory/x_nlp_7.dat};
\addplot [z_best] table [col sep=space] {./trajectory/x_nlp_8.dat};
\end{groupplot}

\end{tikzpicture}

%% file: Chapters/conclusion.tex
We introduced a policy optimization framework for MPC policies that combines model-based information with zeroth-order gradient estimation achieving fast learning transients while guaranteeing convergence to a critical point even under imperfect models. 
This approach is well-suited to settings where accurate modeling is difficult, offering robustness without sacrificing efficiency. 
We demonstrated the effectiveness of our algorithm on a nonlinear quadcopter task, showing that it achieves near-optimal performance while outperforming purely model-based and model-free baselines. 
Future work will focus on strengthening the safety guarantees of the approach.

%% file: Chapters/definability.tex
\begin{definition}[\hspace{1sp}{Definitions 1.4 and 1.5, \cite{coste1999introduction}}]\label{def:definable}
A collection $\mathcal{O}=(\mathcal{O}_n)_{n\in\N}$, where $\mathcal{O}_n \subset 2^{\R^n}$ for each $n\in\N$, is an \textit{o-minimal structure} on $(\R,+,\cdot)$ if: 1) ~all semialgebraic subsets of $\R^n$ belong to $\mathcal{O}_n$; 
2) ~$\mathcal{O}_1$ is the set of all finite unions of points and intervals; 
3) ~$\mathcal{O}_n$ is a boolean subalgebra of $2^{\R^n}$; 
4) ~$A\times B \in \mathcal{O}_{n+m}$ for all $A,B\in \mathcal{O}_n \times \mathcal{O}_m$; 
5) ~the set $\{ v\in\R^n: (v,w) \in A \text{ for some } w\in\R \}$, for any $A\in\mathcal{O}_{n+1}$, belongs to $\mathcal{O}_n$.
A subset of $\R^n$ which belongs to $\mathcal{O}$ is said to be \textit{definable}. 
\end{definition}
A function $\varphi:\R^n\to\R^p$ is \textit{definable} if its graph $\{(x,v):v=\varphi(x)\}$ is definable. 
Definable functions are ubiquitous in control and optimization,
containing, for example, all semialgebraic and globally subanalytic functions.
Moreover, they are stable under most common operations, such as composition, differentiation, and affine transformations.

%% file: Chapters/appendix_mpc_path_diff.tex
We focus on parametric quadratic programs in standard form following \cite{zuliani2023bp}\\[-1.25em]
\begin{subequations}
\noindent\qquad~~
\begin{minipage}[t]{0.4\textwidth}
\begin{align}
\operatorname*{min.}_x &\quad \tfrac{1}{2} x^\top Q(\theta) x + q(\theta)^\top x \notag \\
\text{s.t.} &\quad F(\theta)x = f(\theta), \label{eq:QP} \\
&\quad G(\theta)x \le g(\theta). \notag
\end{align}
\end{minipage}\qquad~~
\begin{minipage}[t]{0.4\textwidth}
\begin{align}
\operatorname*{min.}_{z=(\lambda,\mu)} &\quad \tfrac{1}{2} z^\top H(p) z + h(p)^\top z, \notag \\[-0.2em]
\text{s.t.} &\quad \lambda \ge 0, \qquad\qquad\qquad \label{eq:QP_dual}
\end{align}
\end{minipage}\qquad
\end{subequations}\\

\noindent where $\theta$ is a parameter and, assuming $Q(\theta) \succ 0$, $H(p)$ and $h(p)$ are defined as
\begin{align*}
H(\theta) & = \begin{bmatrix}
G(\theta) Q(\theta)^{-1} G(\theta)^\top & G(\theta) Q(\theta)^{-1} F(\theta)^\top\\
F(\theta) Q(\theta) ^{-1} G(\theta)^\top & F(\theta) Q(\theta) ^{-1} F(\theta)^\top
\end{bmatrix},~
h(\theta) = \begin{bmatrix}
G(\theta)Q(\theta)^{-1}q(\theta)+g(\theta)\\
F(\theta)Q(\theta)^{-1}q(\theta)+f(\theta)
\end{bmatrix}.
\end{align*}
MPC problems with linear (or affine) dynamics and polytopic constraints can always be expressed as in \eqref{eq:QP}.
For systems with nonlinear dynamics, we can obtain a linear prediction model by linearizing at the current state or along the trajectory obtained as solution of the previous MPC problem.
We discussed these ideas thoroughly in \cite[Section VI-A]{zuliani2023bp}.

Given a dual optimizer $z(\theta)$, that is, a solution of \eqref{eq:QP_dual}, one can retrieve the primal optimizer as
\begin{align*}
x(\theta) = \mathcal{G}(z,\theta) := -Q(\theta)^{-1} ( [F(\theta)^\top~G(\theta)^\top] z + q(\theta) ).
\end{align*}
To ensure path-differentiability, we need the following definition.
\begin{definition}\label{def:LICQ}
A minimizer of \eqref{eq:QP} satisfies the \emph{linear independence constraint qualification} (LICQ) if the rows of $G(\theta)$ associated to active constraints and the rows of $F(\theta)$ are linearly independent.
\end{definition}
\begin{theorem}[Theorem 1, \cite{zuliani2023bp}]\label{theorem:diff}
Suppose $x(\theta)$ is a minimizer of \eqref{eq:QP} satisfying LICQ and $Q(\theta) \succ 0$. Then the mapping $\theta \mapsto x(\theta)$ is locally unique, Lipschitz and definable.
Moreover, we have that
\begin{align*}
W+Q(\theta)^{-1}[G(\theta)^\top~F(\theta)^\top]Z \in \J_x(\theta),~~ W \in \J_{\mathcal{G},\theta}(z,\theta)
\end{align*}
where $z$ solves \eqref{eq:QP_dual}, and $Z=U^{-1}V$, with $U \in J_{P_C}(I-\gamma H(\theta))-I$, $V \in -\gamma J_{P_C}(Az+B)$, is an element of the conservative Jacobian of the dual problem \eqref{eq:QP_dual}, with $J_{P_C} = \operatorname*{diag}(\operatorname*{sign}(\lambda),\mathds{1}_{n_\text{eq}})$, and $A \in \J_H(\theta)$, $B \in \J_h(\theta)$.
\end{theorem}
Fulfilling the assumptions of Theorem~\ref{theorem:diff} is not difficult: strong convexity depends on the parameterization of the cost, and it is therefore possible to satisfy this condition by design. LICQ holds in many situations, for example if the constraints are box constraints on the state and input.

%% file: Chapters/ablation.tex
We repeated the experiments in Section~\ref{section:simulation} with different values of $\delta$ and over multiple random seeds (that determine the value of the random perturbations $v_k$ needed to form the zeroth-order estimation). The results are reported in Table~\ref{tab:ablation:delta} (where we run for $1000$ iterations and modified the stepsize to $\rho=\texttt{1e-4}$), and in Table~\ref{tab:ablation:seed}. All costs denote the best attained cost in the last $50$ iterations.
These randomized tests highlight the robustness of the approach.

\begin{table}[ht]
\centering
\caption{Ablation study on $\delta$.}
\small
\label{tab:ablation:delta}
\begin{tabular}{l c c c c c c c}
\toprule
\textbf{$\delta$} & \texttt{1e-04} & \texttt{5e-05} & \texttt{1e-05} & \texttt{5e-06} & \texttt{1e-06} & \texttt{5e-07} & \texttt{1e-07} \\
\midrule
\textbf{Cost} & $1092.9$ & $1093.1$ & $1098.1$ & $1088.9$ & $1097.1$ & $1092.4$ & $1093.8$ \\
\bottomrule
\end{tabular}
\end{table}

\begin{table}[ht]
\centering
\caption{Ablation study on the random seed.}
\small
\label{tab:ablation:seed}
\begin{tabular}{l c c c c c c c c c c}
\toprule
\textbf{Seed} & $0$ & $1$ & $2$ & $3$ & $4$ & $5$ & $6$ & $7$ & $8$ & $9$ \\
\midrule
\textbf{Cost} & $1094.1$ & $1094.1$ & $1092.2$ & $1090.8$ & $1092.8$ & $1090.7$ & $1093.4$ & $1092.5$ & $1104.9$ & $1097.3$ \\
\bottomrule
\end{tabular}
\end{table}

We report the following computation times (in seconds) for a laptop running Windows 11 (32GB RAM, i7-1165G7 2.80GHz): 
the average time to solve each QP was $0.0030$ (variance: $2.8227 \times 10^{-7}$), 
the time to compute the Jacobian was $0.0061$ (variance: $4.4733 \times 10^{-7}$), 
and the total time for each closed-loop iteration was $1.6032$ (variance: $0.0121$).